\documentclass[copyright,creativecommons]{eptcs}
\providecommand{\event}{AFL 2014} % Name of the event you are submitting to
\usepackage[latin1]{inputenc}
\usepackage{amsmath,amssymb,latexsym}
\usepackage{amsthm}
\usepackage{enumerate}
\usepackage{tikz}
\usetikzlibrary{automata}

\newtheorem{theorem}{Theorem}
\newtheorem{corollary}{Corollary}
\newtheorem*{remark}{Remark}
\newtheorem{lemma}{Lemma}
\newtheorem{proposition}{Proposition}
\theoremstyle{definition}
\newtheorem*{definition}{Definition}

\newcommand{\nat}{\mathbb{N}}

\newcommand{\valph}[3]{\langle#1,#2,#3\rangle}
\newcommand{\lsu}{L_{\mathrm{su}}}
\newcommand{\lmwm}{L_{\mathrm{mwm}}}
\newcommand{\wred}{\le_{\mathrm{W}}}
\newcommand{\A}{\mathcal{A}}
\newcommand{\B}{\mathcal{B}}
\newcommand{\qin}{q_0}
\newcommand{\Aless}{\prec}
\newcommand{\prio}{ht}

\newcommand\e{\ensuremath{\epsilon}}
\newcommand\w{\ensuremath{\omega}}
\newcommand{\alphabet}{A}
\newcommand{\alphabetE}{A_\e}

\newcommand{\GammaB}{\Gamma_{\!\bot}}

\newcommand\bbN{\mathbb N}

\newcommand{\BSigma}{{\boldsymbol \Sigma}}
\newcommand{\BPi}{{\boldsymbol \Pi}}

\tikzstyle{stack}=[inner sep=0pt,minimum height=.9cm,draw,rectangle]

\title{Decision Problems for Deterministic Pushdown Automata on Infinite Words}
\author{Christof L\"oding
\institute{Lehrstuhl Informatik 7 \\ RWTH Aachen University\\ Germany}
\email{loeding@cs.rwth-aachen.de}
}
\def\titlerunning{Decision Problems for $\omega$-DPDAs}
\def\authorrunning{C. L\"oding}
\begin{document}
\maketitle

\begin{abstract}
The article surveys some decidability results for DPDAs on infinite
words ($\omega$-DPDA). We summarize some recent results on the
decidability of the regularity and the equivalence problem for the
class of weak $\omega$-DPDAs.  Furthermore, we present some new
results on the parity index problem for $\omega$-DPDAs. For the
specification of a parity condition, the states of the omega-DPDA are
assigned priorities (natural numbers), and a run is accepting if the
highest priority that appears infinitely often during a run is
even. The basic simplification question asks whether one can determine
the minimal number of priorities that are needed to accept the
language of a given $\omega$-DPDA. We provide some decidability results
on variations of this question for some classes of $\omega$-DPDAs.
\end{abstract}

%%%%%%%%%%%%%%%%%%%%%%%%%%%%%%%%%%%%%%%%%%%%%%%%%%%%%%%%%%%%%%%%%%%%%%%%%%%%%%%%
%%%%%%%%%%%%%%%%%%%%%%%%%%%%%%%%%%%%%%%%%%%%%%%%%%%%%%%%%%%%%%%%%%%%%%%%%%%%%%%%
\section{Introduction}
%%%%%%%%%%%%%%%%%%%%%%%%%%%%%%%%%%%%%%%%%%%%%%%%%%%%%%%%%%%%%%%%%%%%%%%%%%%%%%%%
%%%%%%%%%%%%%%%%%%%%%%%%%%%%%%%%%%%%%%%%%%%%%%%%%%%%%%%%%%%%%%%%%%%%%%%%%%%%%%%%
Finite automata, which are used as a tool in many areas of computer
science, have good closure and algorithmic properties. For example,
language equivalence and inclusion are decidable (see
\cite{HopcroftU79}), and for many subclasses of the regular languages
it is decidable whether a given automaton accepts a language inside
this subclass (see \cite{Straubing94} for some results of this
kind). In contrast to that, the situation for pushdown automata is
much more difficult. For nondeterministic pushdown automata, many
problems like language equivalence and inclusion are undecidable (see
\cite{HopcroftU79}), and it is undecidable whether a given
nondeterministic pushdown automaton accepts a regular language. The
class of languages accepted by deterministic pushdown automata forms a
strict subclass of the context-free languages. While inclusion remains
undecidable for this subclass, a deep result from \cite{Senizergues01}
shows the decidability of the equivalence problem. Furthermore, the
regularity problem for deterministic pushdown automata is also
decidable \cite{Stearns67,Valiant75}.

While automata on finite words are a very useful model, some
applications, in particular in verification by model checking (see
\cite{BaierK08}), require extensions of these models to infinite
words. Although the theory of finite automata on infinite words
(called \w-automata in the following) usually requires more
complex constructions because of the more complex acceptance
conditions, many of the good properties of finite automata on finite
words are preserved (see \cite{PerrinP04} for an overview). Pushdown
automata on infinite words (pushdown \w-automata) have been
studied because of their ability to model executions of
non-terminating recursive programs. In \cite{EsparzaHRS00} efficient
algorithms for checking emptiness of B\"uchi pushdown automata are
developed (a B\"uchi automaton accepts an infinite input word if it
visits an accepting state infinitely often during its run). Besides
these results, the algorithmic theory of pushdown \w-automata
has not been investigated very much. For example, in \cite{CohenG78}
the decidability of the regularity problem for deterministic pushdown
\w-automata has been posed as an open question and to our
knowledge no answer to this question is known. Furthermore, it is
unknown whether the equivalence of deterministic pushdown
\w-automata is decidable.

The first part of this article summarizes some recent partial results
on the regularity and equivalence problem for deterministic pushdown
\w-automata from \cite{LodingR12}.

In the second part we consider decision problems concerning the
acceptance condition of the automata. One of the standard acceptance
conditions of \w-automata is the parity condition (see \cite{GTW2002}
for an overview of possible acceptance conditions). Such a condition
is specified by assigning priorities (natural numbers) to the states
of the automaton, using even priorities for ``good'' states and odd
priorities for the ``bad'' states. A run is accepting if among the
states that occur infinitely often the highest priority is even. For
deterministic automata (independent of the precise automaton model),
one can show that more languages can be accepted if more priorities
are used. So the number of priorities required for accepting a
language is a measure for the complexity of the language. A natural
decision problem arising from that, is the question of determining for
a given deterministic parity automaton the smallest number of
priorities that are needed for accepting the language of the
automaton. This referred to as the parity index problem.

For finite deterministic parity automata, the minimal number of
priorities required for accepting the language can be computed in
polynomial time, and a corresponding automaton can be constructed by
simply reassigning priorities in the allowed range to the states of
the given automaton \cite{Carton99}.  For deterministic pushdown
parity automata it was shown in \cite{Linna77} that it is decidable
whether a given automaton is equivalent to a deterministic pushdown
B\"uchi automaton. We present here the general result that the parity
index problem for deterministic pushdown parity automata is
decidable. The method is based on parity games on pushdown graphs and
has already been described in the PhD thesis \cite{RepkeDiss}.

We further consider a model of deterministic pushdown automata in
which the types of the action on the pushdown store are determined by
the input symbols, called visibly pushdown automata (VPA)
\cite{AM2004VPL}.  In these automata, the input alphabet is partitioned
into three sets of symbols, referred to as call, return, and internal
symbols. On reading a call, the pushdown automaton has to add a symbol
to the stack, on reading a return, it has to remove a symbol from the
stack, and on reading an internal, it does not alter the stack. It
turns out that, for a fixed partition of the input alphabet, this
class of automata has good closure and algorithmic properties
\cite{AM2004VPL}. On finite words it is even possible to determinize
such VPAs. However, it turns out that B\"uchi VPAs cannot, in general,
be transformed into equivalent deterministic Muller or parity VPAs
\cite{AM2004VPL}. To resolve this problem, in \cite{LMS2004VPG} a
variation of the parity condition has been proposed, referred to as
stair parity condition. It is defined as a standard parity condition,
however, it is not evaluated on the sequence of all states but only on
the sequence of states that occur on steps of the run. A step is a
configuration in the run such that no later configuration has a
smaller stack height. In \cite{LMS2004VPG} it is shown that each
nondeterministic B\"uchi VPA can be transformed into an equivalent
deterministic stair parity VPA. We prove here that the stair parity
index problem for deterministic VPAs can be solved in polynomial
time. We also consider the question whether a given stair parity VPA
is equivalent to a parity VPA (with a standard parity condition
instead of a stair condition). For the particular case of stair
B\"uchi VPAs we show that this problem is decidable. 

The remainder of this paper is structured as follows. In
Section~\ref{sec:preliminaries} we introduce some basic terminology
and definitions. In Section~\ref{sec:regularity-equivalence} we
consider the regularity and equivalence problem for
$\omega$-DPDAs. Section~\ref{sec:parity-index} is about the parity
index of parity DPDAs and stair parity DVPAs. In Section~\ref{sec:removing-stair} we
show how to decide whether the stair condition is needed for accepting
the language of a given stair B\"uchi DVPAs. In
Section~\ref{sec:conclusion} we give a short conclusion.

%%%%%%%%%%%%%%%%%%%%%%%%%%%%%%%%%%%%%%%%%%%%%%%%%%%%%%%%%%%%%%%%%%%%%%%%%%%%%%%%
%%%%%%%%%%%%%%%%%%%%%%%%%%%%%%%%%%%%%%%%%%%%%%%%%%%%%%%%%%%%%%%%%%%%%%%%%%%%%%%%
\section{Preliminaries} \label{sec:preliminaries}
%%%%%%%%%%%%%%%%%%%%%%%%%%%%%%%%%%%%%%%%%%%%%%%%%%%%%%%%%%%%%%%%%%%%%%%%%%%%%%%%
%%%%%%%%%%%%%%%%%%%%%%%%%%%%%%%%%%%%%%%%%%%%%%%%%%%%%%%%%%%%%%%%%%%%%%%%%%%%%%%%
We denote the set of natural numbers (including $0$) by $\bbN$.
For a set $S$ we denote its cardinality by $|S|$.
Let $\alphabet$ be an alphabet, i.e., a finite set of symbols, then
$\alphabet^*$ is the set of finite words over $\alphabet$, and
$\alphabet^\w$ the set of $\w$-words over $\alphabet$, i.e., infinite
sequences of $\alphabet$ symbols indexed by the natural numbers.  The
subsets of $\alphabet^*$ are called languages, and subsets of
$\alphabet^\w$ are called $\w$-languages. The length of a finite word
$w \in A^*$ is denoted by $|w|$, and the empty word is $\e$.  We
assume the reader to be familiar with regular languages, i.e., the
languages specified by regular expressions or equivalently by finite
state automata (see, for example, \cite{HopcroftU79} for basics on
regular languages).  

We are mainly concerned with deterministic pushdown automata in this
work. We first define pushdown machines, which are pushdown automata
without acceptance condition. We then obtain pushdown automata by
adding an acceptance condition.

A \emph{deterministic pushdown machine} $\mathcal M =
(Q,\alphabet,\Gamma,\delta,q_0,\bot)$ consists of
\begin{itemize}
\item
a finite state set $Q$ and initial state $q_0\in Q$,
\item
a finite input alphabet $\alphabet$ (we abbreviate $\alphabetE
= \alphabet\cup\{\e\}$),
\item
a finite stack alphabet $\Gamma$ and initial stack symbol
$\bot\not\in\Gamma$ (let $\GammaB = \Gamma\cup\{\bot\}$),
\item
a partial transition function $\delta:
Q\times\GammaB\times\alphabetE \to Q\times\GammaB^*$ such that
for each $p \in Q$ and $A \in \GammaB$:
\begin{itemize}
\item $\delta(p,Z,a)$ is defined for all $a \in \alphabet$ and
$\delta(p,Z,\e)$ is undefined, or the other way round.
\item For each transition $\delta(p,Z,a) = (q,W)$ with $a \in \alphabetE$
 the bottom symbol $\bot$ stays at the bottom of
the stack and only there, i.e., $W\in\Gamma^*\bot$ if $Z=\bot$ and
$W\in\Gamma^*$ if $Z\neq\bot$.
\end{itemize}
\end{itemize}
The set of configurations of $\mathcal M$ is $Q\Gamma^*\bot$ where
$q_0\bot$ is the initial configuration. The stack consisting only of
$\bot$ is called the empty stack. A configuration $q\sigma$ is also
written $(q,\sigma)$. For a given input word $w\in\alphabet^*$
or $w\in\alphabet^\w$, a finite resp.\ infinite sequence
$q_0\sigma_0,q_1\sigma_1,\ldots$ of configurations with $q_0\sigma_0=q_0\bot$ is a
run of $w$ on $\mathcal M$ if there are $a_i \in \alphabet_\e$ with
$w=a_1 a_2 \cdots$ and $\delta(q_i,Z,a_{i+1}) = (q_{i+1},U)$ is such
that $\sigma_i=ZV$ and $\sigma_{i+1}=UV$ for some stack suffix $V\in\GammaB^*$.

For finite words, we consider the model of a deterministic
pushdown automaton (DPDA) $\mathcal A = (\mathcal M,F)$ consisting of
a deterministic pushdown machine $\mathcal M =
(Q,\alphabet,\Gamma,\delta,q_0,\bot)$ and a set of final states
$F\subseteq Q$.  It accepts a word $w\in\alphabet^*$ if $w$
induces a run ending in a final state.  These words form the language
$L_*(\mathcal A)\subseteq\alphabet^*$.
For \w-words, we consider two types of acceptance conditions, namely
B\"uchi and parity conditions.  A B\"uchi DPDA $\mathcal A = (\mathcal
M,F)$ is specified in the same way as a DPDA on finite words. The
$\w$-language $L_\w(\A)$ defined by $\A$ is the set of all $\w$-words
$w$ for which the run of $\A$ on $w$ contains a state from $F$ at
infinitely many positions.

For a parity DPDA, the acceptance condition is specified by a function
$\Omega: Q \rightarrow \bbN$, which assigns a number to each state,
which is referred to as its priority. A run is accepting if the highest
priority that occurs infinitely often is even. Note that B\"uchi
conditions can be specified as parity conditions by assigning priority
2 to states in $F$ and priority 1 to states outside $F$.

In Section~\ref{sec:regularity-equivalence} we consider the class of
weak DPDAs. These are parity DPDAs, in which the transitions can never
lead from one state $q$ to another state $q'$ with a smaller
priority. Hence, in a run of a weak DPDA the sequence of priorities is
monotonically increasing, which implies that the sequence is
ultimately constant. It follows that each weak DPDA is equivalent to
the B\"uchi DPDA that uses the set of states with even priority as set
of final states. We therefore also use term weak B\"uchi DPDAs to
emphasize that it is a subclass of B\"uchi DPDAs.

In general, we refer to DPDAs on infinite words as $\w$-DPDAs if we do
not explicitly specify the type of acceptance. For simplicity, we
assume that infinite sequences of $\e$-transitions are not possible in
$\w$-DPDAs. Such sequences can be eliminated by redirecting certain
$\e$-transitions into corresponding sink states (the acceptance status
of such a state would depend on the exact semantics one uses for runs
that end in an infinite $\e$-sequence). It is sufficient to compute
the pairs $(q,Z)$ of states $q$ and top stack symbols $Z$ such that
there is a run of $\e$-transitions leading from $qZ\bot$ to some
configuration of the form $qZWZ\bot$, such that the $Z$ at the bottom
of the stack is never removed during this run. These pairs can be
computed efficiently (see \cite{EsparzaHRS00}), and it is not difficult
to see that redirecting the $\e$-transitions from these pairs $(q,Z)$
is sufficient for eliminating all infinite $\e$-sequences.

%% We assume that 
%% To ensure that an infinite run reads an infinite word, we require that
%% $\mathcal A$ has no infinite sequence of \e-transitions.  The presence
%% of such sequences can be tested and they can be removed in polynomial
%% time by redirecting some of the \e-transitions into sink states.%
%% Under this assumption, we define $\delta_*(w)$ for a finite word $w$
%% to be the last state $q_n$ of a run $q_0\sigma_0,\ldots,q_n\sigma_n$ on $w$ such
%% that there is no further \e-transition possible.

% Visibly pushdown automata
We also consider the model of deterministic visibly pushdown automata
(DVPA) \cite{AM2004VPL}. These automata are defined with respect to a
partitioned alphabet $\alphabet=\alphabet_c \cup \alphabet_i \cup
\alphabet_r$, where $\alphabet_c$ contains all letters that can only
occur in transitions pushing some symbol onto the stack (call
symbols), $\alphabet_r$ those forcing the automaton to pop a symbol
from the stack (return symbols), and $\alphabet_i$ those leaving the
stack unchanged (internal symbols). Furthermore, DVPAs do not have
$\e$-transitions. We also adopt the general convention that VPAs do
not consider the top-most stack symbol in their transitions. This
simplifies several arguments. We can make this assumption without loss
of generality, because it is possible to always keep track of the
top-most stack symbol in the control state.

Formally, a deterministic visibly pushdown machine over the partitioned
alphabet $\alphabet=\alphabet_c \cup \alphabet_i \cup \alphabet_r$ is
of the form $\mathcal M = (Q,\alphabet,\Gamma,\delta,q_0,\bot)$, where
$\delta$ consists of three transition functions
\[
\begin{array}{l}
\delta_c: Q \times \alphabet_c \rightarrow Q \times \Gamma \\
\delta_r: Q \times \Gamma \times \alphabet_r \rightarrow Q  \\
\delta_i: Q \times \alphabet_i \rightarrow Q \\
\end{array}
\]
Instead of defining the semantics of these transitions directly, we
simply describe how the corresponding transitions in a standard DPDA
would look like.  A call transition $\delta_c(q,c) = (p,Z)$
corresponds to a set of transitions $\delta(q,Y,c) = (p,ZY)$ for each
$Y \in \GammaB$. A return transition $\delta_r(q,Z,r) = p$ corresponds
to the transition $\delta_r(q,Z,r) = (p,\e)$, and an internal
transition $\delta_i(q,i) = p$ to a set of transitions $\delta(q,Y,i)
= (p,Y)$ for each $Y \in \GammaB$.  Note that this definition does not
admit transitions for return symbols on the empty stack. In
\cite{AM2004VPL} such transitions are possible, but we prefer to use
the simpler model here to ease the presentation.

By adding an acceptance condition, we obtain DVPAs as in the general
case. As for $\omega$-DPDAs, we are interested in $\omega$-DVPAs with
B\"uchi or parity condition. However, we also consider a variant of
the parity condition referred to as stair parity condition
\cite{LMS2004VPG}. The condition is specified in the same way as
before, however, it is evaluated only on a subsequence of the run,
namely on the sequence of steps, as defined below.

A configuration $q\sigma$ in a run of a DVPA $\A$ is called a step if
the stack height of all configurations $q'\sigma'$ that come later in
the run is bigger than the stack height of $q\sigma$, i.e., $|\sigma|
\le |\sigma'|$. Note that the positions of the steps do not depend
on the automaton, but only on the input word, because the type of the
stack operation is determined for each input symbol. We can now define
stair visibly pushdown automata. The only difference to visibly
pushdown automata is that they evaluate the acceptance condition only
for the subsequence of the run containing consisting of the steps.

In other words, a stair parity DVPA has the same components as a
parity DVPA. An input is accepted if in the run on this input the
maximal priority that occurs infinitely often on a step is even. In
the same way we obtain stair B\"uchi DVPAs, which accept if an
accepting state occurs on infinitely many steps.

We end this section by introducing some more terminology for visibly
pushdown automata that is used in Sections~\ref{sec:parity-index} and
\ref{sec:removing-stair}. 

The set of well matched words over $\alphabet=\alphabet_c \cup
\alphabet_i \cup \alphabet_r$ is, intuitively speaking, the set of
well-balanced words in which for each position with a call symbol
there is a later position at which this call is ``closed'' by some
return symbol (and vice versa, each return position has a
corresponding previous call position). Formally, the set is defined
inductively as follows:
\begin{itemize}
\item Each $a \in \alphabet_i$ is a well matched word.
\item If $u$ and $v$ are well-matched words, then $uv$ is a well
  matched word.
\item If $w$ is a well matched word, then $cwr$ is a well-matched word
  for each $c \in \alphabet_c$ and each $r \in \alphabet_r$.
\end{itemize}

The words that are created by the last rule are referred to as
minimally well-matched words. Let $\lmwm$ denote this set, i.e., the
words of the form $cwr$ with a call $c$, a return $r$, and a
well-matched word~$w$.

The canonical language that can be accepted by a stair B\"uchi DVPA
but by no parity DVPA is the language $\lsu$ of strictly unbounded
words, containing all words over $\valph{\{c\}}{\emptyset}{\{r\}}$
with an infinite number of unmatched calls. More formally, an infinite
word is in $\lsu$ if it is of the form $w_1cw_2cw_3c \cdots$ for
well-matched words $w_i$. In \cite{AM2004VPL} it is shown that $\lsu$
cannot be accepted by a parity DVPA. But it is easy to construct a
stair B\"uchi DVPA $\A$ for $\lsu$ using only a single stack symbol
and one accepting and one non-accepting state (see \cite{LMS2004VPG}),
where $\A$ moves into the accepting state for each $c$, and into the
non-accepting state for each $r$. Note that the position after reading
a $c$ is a step in the run iff this $c$ does not have a matching
return. Thus, there are infinitely many unmatched calls iff there are
infinitely many accepting states on steps.

%%%%%%%%%%%%%%%%%%%%%%%%%%%%%%%%%%%%%%%%%%%%%%%%%%%%%%%%%%%%%%%%%%%%%%%%%%%%%%%%
%%%%%%%%%%%%%%%%%%%%%%%%%%%%%%%%%%%%%%%%%%%%%%%%%%%%%%%%%%%%%%%%%%%%%%%%%%%%%%%%
\section{Regularity and Equivalence} \label{sec:regularity-equivalence}
%%%%%%%%%%%%%%%%%%%%%%%%%%%%%%%%%%%%%%%%%%%%%%%%%%%%%%%%%%%%%%%%%%%%%%%%%%%%%%%%
%%%%%%%%%%%%%%%%%%%%%%%%%%%%%%%%%%%%%%%%%%%%%%%%%%%%%%%%%%%%%%%%%%%%%%%%%%%%%%%%
In this section we summarize results from \cite{LodingR12} that show
how to solve the regularity problem and the equivalence problem for
weak $\omega$-DPDAs. The proof uses a reduction to the corresponding
problems for DPDAs on finite words. More details on these results can
be found in \cite{LodingR12} and in \cite{RepkeDiss}.

The regularity problem for DPDA is the problem of deciding for a given
DPDA whether it accepts a regular language. It has been shown to be
decidable in \cite{Stearns67} and the complexity has been improved in
\cite{Valiant75}.
\begin{theorem}[\cite{Stearns67}] \label{the:DPDA-regularity}
The regularity problem for DPDAs is decidable.
\end{theorem}
The rough idea of the proof is as follows. Assuming that the language
of the given DPDA is regular, one shows that for each configuration
above a certain height (depending on the size of the DPDA), there is
an equivalent configuration of smaller height. A finite state machine
can then be constructed by redirecting the transitions into higher
configurations to their equivalent smaller counterparts.  Here, two
configurations are considered to be equivalent if they define the same
language when considered as initial configuration of the DPDA. The
decision method for the regularity problem is then based on the
characterization of the regular languages in terms of the
Myhill/Nerode equivalence. For a language $L \subseteq \alphabet^*$, the
Myhill/Nerode equivalence is defined as follows for words $u,v \in
\alphabet^*$:
\[
u \sim_L v \mbox{ iff } \forall w \in \alphabet^*\,:\; uw \in L
\Leftrightarrow vw \in L.
\]
A language of finite words is regular if, and only if, it has finitely
many Myhill/Nerode equivalence classes, and these classes can be used
as states for a canonical finite automaton for the language.

Unfortunately, a corresponding result is not true for $\omega$-regular
languages, in general. However, the subclass of weak $\omega$-regular
languages possesses a similar characterization in terms of an
equivalence \cite{Staiger83}. This similarity raises the question
whether the decidability results for DPDAs on finite words can be
lifted to weak DPDAs on infinite words.

In \cite{LodingR12} it is shown that this is indeed possible. In fact,
it is even possible to reduce questions for weak $\omega$-DPDAs to DPDAs
on finite words. To establish such a connection, we associate a
language $L_*(\A)$ of finite words to a weak $\omega$-DPDA $\A$, which
is obtained by viewing $\A$ as a DPDA on finite words and taking the
set of states with an even priority as the set of final states. 

The first attempt for reducing the regularity problem for weak
$\omega$-DPDAs to the regularity problem for DPDAs would be to test
$L_*(\A)$ for regularity, where $\A$ is the given weak
$\omega$-DPDA. This approach is sound because regularity of $L_*(\A)$
implies $\omega$-regularity of $L_\omega(\A)$: a finite deterministic
automaton for $L_*(\A)$ viewed as a B\"uchi automaton defines
$L_\omega(\A)$ because it visits final states at the same positions
as~$\A$.

That the approach is not complete is illustrated by the following
simple example. Consider the alphabet $\{a,b\}$ and the
$\omega$-language $a^*b^\omega$ of words starting with a finite
sequence of $a$ followed by an infinite sequence of $b$. Obviously,
this language is regular. A weak $\omega$-DPDA $\A$ could proceed as
follows to accept this language. It starts by pushing a symbol onto
the stack for each $a$. When the first $b$ comes in the input, it
changes its state and starts popping the stack symbols again. Once the
bottom of the stack is reached, it changes to an accepting state and
remains there as long as it reads further $b$ (if another $a$ comes,
then the input is rejected). Since the finite $a$-sequence is followed
by infinitely many $b$, it is guaranteed that $\A$ reaches the
accepting state if the input is from $a^*b^\omega$. Note that this is
a weak $\omega$-DPDA because it can change once from non-accepting to
accepting states, and once more back to non-accepting states.  The
language $L_*(\A)$ of this weak $\omega$-DPDA is the set of all finite
words of the form $a^mb^n$ with $n \ge m$ because $\A$ reaches the
accepting state only after it has read as many $b$ as $a$. Thus,
$L_*(\A)$ is non-regular although $L_\omega(\A)$ is.

For this example, the problem would be solved if $\A$ switches to an
accepting state as soon as the first $b$ is read (instead of deferring
this change to the stack bottom).  In general, one can
show that each weak $\omega$-DPDA can be transformed in such a way
that the above reduction to the regularity test for $L_*(\A)$, as shown
be the following theorem. 

\begin{theorem}[\cite{LodingR12}] \label{the:weak-parity-reduction}
There is a normal form for weak $\omega$-DPDAs with the following
properties:
\begin{enumerate}
\item For a weak $\omega$-DPDA $\A$ in normal form, the language
  $L_\omega(\A)$ is $\omega$-regular if, and only if, $L_*(\A)$ is regular.
\item Given two weak $\omega$-DPDAs $\A$ and $\B$ in normal form,
  $L_\omega(\A) = L_\omega(\B)$ if, and only if, $L_*(\A) = L_*(\B)$.
\end{enumerate}
\end{theorem}
Combining the first part of Theorem~\ref{the:weak-parity-reduction}
with Theorem~\ref{the:DPDA-regularity}, we get the decidability of the
regularity problem for weak $\omega$-DPDAs.
\begin{corollary}[\cite{LodingR12}]
The regularity problem for weak $\omega$-DPDAs is decidable.
\end{corollary}
The second part of the theorem can be used to show the decidability of
the equivalence problem for weak $\omega$-DPDAs, based on the
corresponding deep result for DPDAs.
\begin{theorem}[\cite{Senizergues01}] \label{the:DPDA-equivalence}
The equivalence problem for DPDAs is decidable.
\end{theorem}
\begin{corollary}[\cite{LodingR12}]
The equivalence problem for weak $\omega$-DPDAs is decidable.
\end{corollary}

The two problems for the full class of $\omega$-DPDAs remain open. In
\cite{RepkeDiss} a congruence for $\omega$-languages is identified
that characterizes regularity within the class of $\omega$-DPDA
recognizable languages (a language accepted by an $\omega$-DPDA is
regular if, and only if, this congruence has finitely many equivalence
classes). This might be step towards a solution for the regularity
problem. However, the decidability of characterizing criterion remains
open. 

%%%%%%%%%%%%%%%%%%%%%%%%%%%%%%%%%%%%%%%%%%%%%%%%%%%%%%%%%%%%%%%%%%%%%%%%%%%%%%%%
%%%%%%%%%%%%%%%%%%%%%%%%%%%%%%%%%%%%%%%%%%%%%%%%%%%%%%%%%%%%%%%%%%%%%%%%%%%%%%%%
\section{The Parity Index Problem} \label{sec:parity-index}
%%%%%%%%%%%%%%%%%%%%%%%%%%%%%%%%%%%%%%%%%%%%%%%%%%%%%%%%%%%%%%%%%%%%%%%%%%%%%%%%
%%%%%%%%%%%%%%%%%%%%%%%%%%%%%%%%%%%%%%%%%%%%%%%%%%%%%%%%%%%%%%%%%%%%%%%%%%%%%%%%

In this section we are interested in the problem of reducing the
number of priorities used in a parity condition.  Formally, we
consider the following problem. Given a parity DPDA (or stair parity
DVPA) $\A$, compute the smallest number of priorities required for
accepting $L_\omega(\A)$ with a parity DPDA (or stair parity DVPA).  We refer
to these two variants of the problem as the parity index problem for
DPDAs, and the stair parity index problem for stair parity DVPAs.

For finite parity automata, it suffices to change the priority
assignment, in order to obtain an equivalent automaton with the fewest
number of priorities, and this modified priority function can be
computed in polynomial time \cite{Carton99}. 

For parity DPDAs the situation is different, as illustrated by the
example in Figure~\ref{fig:dvpa-parity-example} (taken from
\cite{S2006DA}). We use a DVPA in the example, where $c_1,c_2$ are
calls, $r_1,r_2$ are returns, $i_1,i_2$ are internals, and $Z_1, Z_2$
are stack symbols.  The transitions on call symbols are annotated with
the stack symbol to be pushed, and for the return symbols with the
stack symbol to be popped.  The priority function of the DVPA on the
left-hand side of Figure~\ref{fig:dvpa-parity-example} (indicated as
labels of the states) is minimal for the state set and the transition
structure.  The problem is caused by the state $q_1$, which is part of
the loop in the upper and the lower branch. However, there is no run
of the automaton that traverses both the upper and the lower
branch. If the first symbol in the input is $c_1$, then the automaton
stores $Z_1$ on the stack. Whenever the automaton reaches $q_1$ in the
future, $Z_1$ will be on top of the stack and the automaton can only
use the top branch. For the lower branch and $c_2$ as the first input
symbol the situation is similar.

Splitting $q_1$ into two copies as done in the DVPA on the right-hand
side of the figure, makes it possible to reassign priorities without
using priority $3$. 

\begin{figure}
\begin{tikzpicture}[node distance=2cm,initial text=,scale=.5]
\tikzstyle{every state}=[inner sep=1pt,minimum size=4mm]
\node[state,initial](q0) {$q_0/0$};
\node[state](q1) [right of=q0]{$q_1/2$};
\node[state](q2) [above right of=q1]{$q_2/1$};
\node[state](q3) [right of=q2]{$q_3/0$};
\node[state](q4) [below right of=q1]{$q_4/3$};
\path[->] (q0) edge node[above]{$\scriptstyle c_1/Z_1$} 
                    node[below]{$\scriptstyle c_2/Z_2$} (q1)
          (q1) edge[bend left] node[left]{$\scriptstyle r_1/Z_1$} (q2)
               edge node[right]{$\scriptstyle r_2/Z_2$} (q4)
               edge[loop right] node[right]{$\scriptstyle i_1$} ()
          (q2) edge[bend left] node[above]{$\scriptstyle i_1$} (q3)
               edge node[right]{$\scriptstyle c_1/Z_1$} (q1)
          (q3) edge[loop right] node[right]{$\scriptstyle i_2$} ()
               edge[bend left] node[below]{$\scriptstyle i_1$} (q2)
          (q4) edge[bend left] node[left]{$\scriptstyle c_2/Z_2$} (q1)
;
\end{tikzpicture}%
\begin{tikzpicture}[node distance=2cm,initial text=,scale=.5]
\tikzstyle{every state}=[inner sep=1pt,minimum size=4mm]
\node[state,initial](q0) {$q_0/0$};
\node[state](q1) [above right of=q0]{$q_1/2$};
\node[state](q1') [below right of=q0]{$q_1'/0$};
\node[state](q2) [right of=q1]{$q_2/1$};
\node[state](q3) [right of=q2]{$q_3/0$};
\node[state](q4) [right of=q1']{$q_4/1$};
\path[->] (q0) edge node[above left]{$\scriptstyle c_1/Z_1$} (q1)
               edge node[below left]{$\scriptstyle c_2/Z_2$} (q1')
          (q1) edge[bend left] node[above]{$\scriptstyle r_1/Z_1$} (q2)
               edge[loop below] node[below]{$\scriptstyle i_1$} ()
          (q1')edge[bend left] node[above]{$\scriptstyle r_2/Z_2$} (q4)
               edge[loop above] node[above]{$\scriptstyle i_1$} ()
          (q2) edge[bend left] node[above]{$\scriptstyle i_1$} (q3)
               edge[bend left] node[below]{$\scriptstyle c_1/Z_1$} (q1)
          (q3) edge[loop right] node[right]{$\scriptstyle i_2$} ()
               edge[bend left] node[below]{$\scriptstyle i_1$} (q2)
          (q4) edge[bend left] node[below]{$\scriptstyle c_2/Z_2$} (q1')
;
\end{tikzpicture}

\caption{On the left-hand side: DVPA with minimal number of priorities
  for the given transition structure; on the right-hand side:
  equivalent DVPA with less priorities\label{fig:dvpa-parity-example}}

\end{figure}
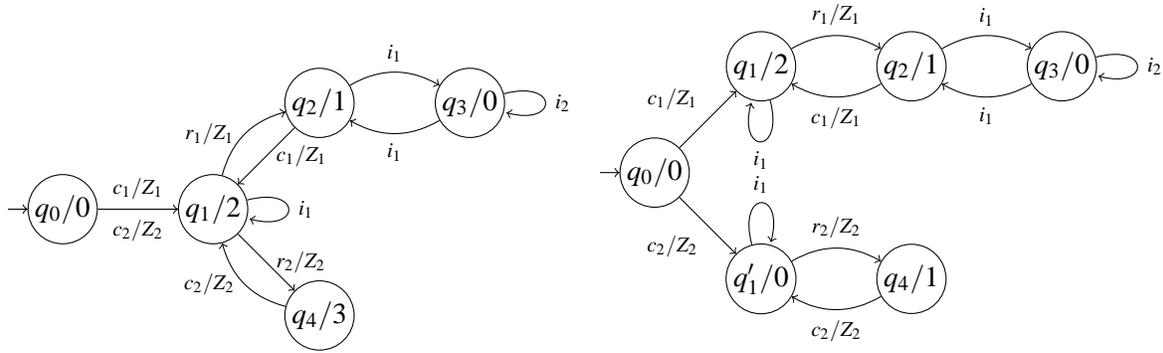

The example illustrates that we need to take a different approach for
computing the parity index of pushdown automata. This approach is also
described in \cite{RepkeDiss}. 

Let $P \subset \bbN$ be a finite set of priorities. A parity DPDA
using only priorities from $P$ is referred to as a $P$-parity DPDA. To
decide whether a given parity DPDA $\A$ has an equivalent $P$-parity
DPDA, consider the following game. There are two players, referred to
as Automaton and Classifier. Automaton starts in the initial
configuration of $\A$ and plays transitions of $\A$. After each move of
Automaton, Classifier chooses one priority from $P$. The idea is that
the classifier wants to prove that there is a $P$-parity DPDA that
accepts $L_\omega(\A)$. If Classifier chooses priority $k$ in a move,
this can be interpreted as ``the parity DPDA that I have in mind would
now be in a state with priority $k$''.

This game can be formalized as a game over a pushdown graph (basically,
the configuration graph of $\A$ enriched by the bounded number of
choices for Classifier). The winning condition states that an infinite
play is won by classifier if, and only if, the two priority sequences,
one induced by the configurations chosen by Automaton, the other given
by the choices of Classifier, are either both accepting or both
rejecting. We refer to this game as the classification game for $\A$
and $P$. The following result can be shown based on results for
computing winning strategies in pushdown games \cite{Walukiewicz01}.

\begin{lemma}
Classifier has a winning strategy in the classification game for $\A$
and $P$ if, and only if, there is $P$-parity DPDA accepting
$L_\omega(\A)$. 
\end{lemma}
For the proof it suffices to observe the following things. If there is
a $P$-parity DPDA $\B$ accepting $L_\omega(\A)$, then Classifier can
simulate the run of $\B$ on the inputs played by Automaton, and always
choose the priority of the current state of $\B$. This obviously
defines a winning strategy because $\A$ and $\B$ accept the same
language. For the other direction one uses the fact that a winning
strategy for Classifier can be implemented by a pushdown automaton
that reads the moves of Automaton and outputs the moves of Classifier
\cite{Walukiewicz01,Fridman10}. This pushdown automaton for the
strategy can easily be converted into $P$-parity DPDA for $L_\omega(\A)$.

For a given parity DPDA there are only finitely many sets $P$ with
less priorities than $\A$ uses. Since it is decidable which player has
a winning strategy in the classification game \cite{Walukiewicz01}, we
obtain an algorithm for solving the parity index problem for DPDAs.
\begin{theorem}
There is an algorithm solving the parity index problem for parity
DPDAs.
\end{theorem}

%%%%%%%%%%%%%%%%%%%%%%%%%%%%%%%%%%%%%%%%%%%%%%%%%%%%%%%%%%%%%%%%%%%%%%%%%%%%%%%%
%%%%%%%%%%%%%%%%%%%%%%%%%%%%%%%%%%%%%%%%%%%%%%%%%%%%%%%%%%%%%%%%%%%%%%%%%%%%%%%%
\subsection*{Stair Parity Index}
%%%%%%%%%%%%%%%%%%%%%%%%%%%%%%%%%%%%%%%%%%%%%%%%%%%%%%%%%%%%%%%%%%%%%%%%%%%%%%%%
%%%%%%%%%%%%%%%%%%%%%%%%%%%%%%%%%%%%%%%%%%%%%%%%%%%%%%%%%%%%%%%%%%%%%%%%%%%%%%%%
We now turn to the stair parity index problem for stair parity
DVPAs. In fact, it is possible to use the same game-based approach
because pushdown games with stair conditions can be solved
algorithmically \cite{LMS2004VPG}. However, for stair parity VPAs one
can also adapt the much simpler solution for computing the parity
index of finite parity automata. Note that in the example from
Figure~\ref{fig:dvpa-parity-example} the ``critical'' state $q_1$ can
never occur on a step (moving out of $q_1$ requires to read a return
and thus to pop a symbol). Thus, the priority of $q_1$ is not
important in a stair parity acceptance condition. It turns out that
this is not a coincidence. The result presented below has been
obtained in collaboration with Philipp Stephan, see \cite{S2006DA}.

Consider the transformation graph of a stair parity DVPA $\A$ defined
as follows. The vertices are the states of $\A$. An edge from $q_1$ to
$q_2$ indicates that $q_1$ and $q_2$ can occur on successive steps in
a run of $\A$. An input connecting two successive steps of a run is
either an internal symbol or a minimally well-matched word. Therefore,
this transformation graph can be computed inductively based on the
definition of well-matched words from
Section~\ref{sec:preliminaries}. One starts with the graph containing
only the edges for the internal symbols. In each iteration one
computes the transitive closure of the current graph. Denote this
transitive closure by $T$. Then one checks whether there are
transitions $\delta(q,c) = (q',Z)$ and $\delta(p',r,Z) = p$ for a call
$c$, a return $r$, and a stack symbol $Z$, such that $(q',p') \in
T$. In this case we add the edge $(q,p)$ to the graph. We repeat this
procedure until no more edges are added.

The paths through the transformation graph correspond to the possible
sequences of states on steps in runs of $\A$. We now use the algorithm
from \cite{Carton99} to compute the minimal number of priorities
required on this transformation graph, simply by viewing it as the
transition graph of a finite state deterministic parity automaton. The
resulting assignment of priorities is then also minimal for the stair
parity DVPA $\A$.

\begin{theorem}
The stair parity index problem for stair parity DVPAs can be solved in
polynomial time.
\end{theorem}

\newpage
%%%%%%%%%%%%%%%%%%%%%%%%%%%%%%%%%%%%%%%%%%%%%%%%%%%%%%%%%%%%%%%%%%%%%%%%%%%%%%%%
%%%%%%%%%%%%%%%%%%%%%%%%%%%%%%%%%%%%%%%%%%%%%%%%%%%%%%%%%%%%%%%%%%%%%%%%%%%%%%%%
\section{Removing the Stair Condition} \label{sec:removing-stair}
%%%%%%%%%%%%%%%%%%%%%%%%%%%%%%%%%%%%%%%%%%%%%%%%%%%%%%%%%%%%%%%%%%%%%%%%%%%%%%%%
%%%%%%%%%%%%%%%%%%%%%%%%%%%%%%%%%%%%%%%%%%%%%%%%%%%%%%%%%%%%%%%%%%%%%%%%%%%%%%%%
The goal is to decide for a given stair parity DVPA whether there is an
equivalent parity DVPA and to construct one if it exists. We show
how to decide this problem in general for stair B\"uchi DVPAs. We
comment on the full class of stair parity DVPAs at the end of this
section.

In Section~\ref{sec:preliminaries} we described the language $\lsu$ of
strictly unbounded words over $\valph{\{c\}}{\emptyset}{\{r\}}$,
containing all words with an infinite number of unmatched calls. This
language can be accepted by a stair B\"uchi DVPA but not by a parity
DVPA \cite{AM2004VPL}.  We show that a language $L$ accepted by a
stair B\"uchi DVPA can
\begin{itemize}
\item either be accepted by a parity DVPA, or 
\item $L$ is at least as complex as $\lsu$.
\end{itemize}
To formalize the notion of ``as complex as $\lsu$'', we need to
introduce some terminology and results concerning the topological
complexity of $\omega$-languages.

We can view $\alphabet^\omega$ as a topological space by equipping it
with the Cantor topology, where the open sets are those of the form $L
\alphabet^\omega$ for $L \subseteq \alphabet^*$. Starting from the open sets
one defines the finite Borel hierarchy as a sequence $\BSigma_1,
\BPi_1, \BSigma_2, \BPi_2, \ldots$ of classes of $\omega$-languages as
follows (we omit the finite and only refer to this hierarchy as Borel
hierarchy in the following):
\begin{itemize}
\item $\BSigma_1$ consists of the open sets.
\item $\BPi_i$ consists of the complements of the languages in
  $\BSigma_i$.
\item $\BSigma_{i+1}$ consists of countable unions of  languages in
  $\BPi_i$. 
\end{itemize}
If we denote by $B(\BSigma_i)$ the closure of $\BSigma_i$ under finite
Boolean combinations, then we obtain the following relation between
the classes of the Borel hierarchy, where an arrow indicates strict
inclusion of the corresponding classes:
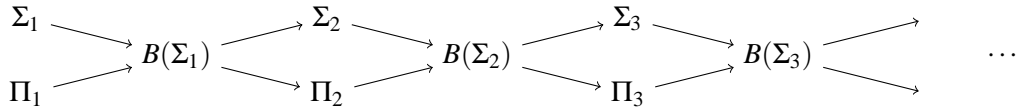
\begin{center}
\begin{tikzpicture}[node distance=2cm]
\path node(s1){$\BSigma_1$} 
  -- +(0,-1)  node(p1){$\BPi_1$}
  -- +(2,-.5) node(b1){$B(\BSigma_1)$}
  -- ++(4,0)  node(s2){$\BSigma_2$}
  -- +(0,-1)  node(p2){$\BPi_2$}
  -- +(2,-.5) node(b2){$B(\BSigma_2)$}
  -- ++(4,0)  node(s3){$\BSigma_3$}
  -- +(0,-1)  node(p3){$\BPi_3$}
  -- +(2,-.5) node(b3){$B(\BSigma_3)$}
  -- ++(4,0)  node(s4){}
  -- +(0,-1)  node(p4){}
  -- +(1,-.5) node(b4){$\cdots$}
;
\path[->] (s1) edge (b1)
          (p1) edge (b1)
          (b1) edge (s2) edge (p2)
          (s2) edge (b2)
          (p2) edge (b2)
          (b2) edge (s3) edge (p3)
          (s3) edge (b3)
          (p3) edge (b3)
          (b3) edge (s4) edge (p4)
 ;
\end{tikzpicture}
\end{center}
The above statement of a language $L$ being at least as complex as
$\lsu$ refers to the topological complexity. It is known that
languages accepted by deterministic automata (independent of the
specific automaton model) with a parity condition are included in
$B(\BSigma_2)$, and in \cite{LMS2004VPG} it is shown that languages
accepted by stair parity DVPAS are in $B(\BSigma_3)$.  Furthermore, it
is known that $\lsu$ is a true $\BSigma_3$-set (it is complete for
$\BSigma_3$ for the reduction notion introduced below)
\cite{CDT2002a}. In particular, it is not contained in $B(\BSigma_2)$.

In our decidability proof we show that specific patterns in a stair
parity DVPA induce a high topological complexity of the accepted
language (namely being at least as complex as $\lsu$). On the other
hand side, the absence of these patterns allows for the construction
of an equivalent parity DVPA.

Before we introduce these patterns, we define the reducibility
notion. Originally, it is defined using continuous functions. For our
purposes it is easier to work with a different definition based on the
Wadge game \cite{Wadge84} (see also \cite{CDT2002a}). 

Consider two alphabets $\alphabet_1,\alphabet_2$ and let $L_1
\subseteq \alphabet_1^\omega$ and $L_2 \subseteq
\alphabet_2^\omega$. The Wadge game $W(L_1,L_2)$ is played between
Players I and II as follows. In each round Player I plays an element
of $\alphabet_1$ and Player II replies with a finite word from
$\alphabet_2^*$ (the empty word is also possible). In the limit,
Player I plays an infinite word $x$ over $\alphabet_1$, and Player II
a finite or infinite word $y$ over $\alphabet_2$. Player II wins if
$y$ is infinite and $x \in L_1$ iff $y \in L_2$.

We write $L_1 \wred L_2$ if Player II has a winning strategy in
$W(L_1,L_2)$. The following theorem is a consequence of basic
properties of $\wred$.
\begin{theorem}[\cite{Wadge84}]\label{the:wadge}
If $L_1 \wred L_2$, then each class of the Borel hierarchy that
contains $L_2$ also contains $L_1$. 
\end{theorem}
We use the following consequence of Theorem~\ref{the:wadge} and the
properties of $\lsu$. 
\begin{lemma}\label{lem:lsu-reduction}
If $\lsu  \wred L$, then $L$ cannot be accepted by a parity DVPA.
\end{lemma}
\begin{proof}
As mentioned above, the languages that can be accepted by parity DPDAs
are contained in $B(\BSigma_2)$. We sketch the proof of this folklore
result for completeness: We apply Theorem~\ref{the:wadge} using the
following argument. Let $\A$ be a parity DPDA and let $P$ be the set
of priorities used by $\A$. Let $L_P \subseteq P^\omega$ be the
sequences of priorities that satisfy the parity condition. Then
$L_\omega(\A) \wred L_P$ because in the Wadge game Player~II can
simply keep track of the run of $\A$ on the word played by Player~I,
and play the corresponding priorities of the states of $\A$. Then
clearly the word played by I is in $L_\omega(\A)$ iff the priority
sequence of II satisfies the parity condition. Now, $L_P$ is easily
seen to be a Boolean combination of $\BSigma_2$-sets.

Since $\lsu$ is not contained in $B(\BSigma_2)$ \cite{CDT2002a}, we
conclude from Theorem~\ref{the:wadge} that $\lsu  \wred L$ implies
that $L$ cannot be accepted by a parity DVPA.
\end{proof}

%------------------------------------------------------------
\paragraph{Forbidden patterns.}
%------------------------------------------------------------
Fix a stair B\"uchi DVPA $\A = (Q,\alphabet, \Gamma, \qin, \delta, F)$
and let $L = L_\omega(\A)$. Recall that $L$ does not contain words with
unmatched returns.  We assume that all states of $\A$ are reachable.

For an input word $u$, states $q,q'$, and stack contents $\sigma,
\sigma'$ we write $(q,\sigma) \xrightarrow{u} (q',\sigma')$ if there
is a run for the input $u$ from $(q,\sigma)$ to $(q',\sigma')$. The
notation $(q,\sigma) \xrightarrow[F]{u} (q',\sigma')$ means that at
least one state from $F$ occurs on a step in this run (for steps to be
defined we assume that all prefixes of $u$ are of non-negative stack
height). Dual to that we write $(q,\sigma) \xrightarrow[\notin F]{u}
(q',\sigma')$ to indicate that no state from $F$ occurs on a step in
this run. If we omit the input word $u$ then this means that there
exists some input word.

It is not difficult to see that $\lsu \wred L$ if there are words $u$
and $u'$, a stack content $\sigma$, and a state $q \in Q \setminus F$
such that
\[
(q,\bot) \xrightarrow[F]{u} (q,\sigma) \xrightarrow{u'}
(q,\bot)
\]
and no final state occurs on steps in this run (in a run that starts
and ends in the empty stack, the steps are the configurations with
empty stack). To prove $\lsu \wred L$, the corresponding winning
strategy for  Player~II in the Wadge game is: $c \mapsto u$ and $r
\mapsto u'$.

Unfortunately, the above condition is not necessary for $\lsu \wred
L$. Consider the stair B\"uchi DVPA $\A$ shown in
Figure~\ref{fig:DVPA-with-forbidden-pattern} with one call symbol $c$
and two return symbols $r_1,r_2$ (the initial state does not matter).
In this automaton the simple pattern described above cannot occur
because the only non-final states are $q$ and $q'$. For these two
states, words $u$ and $u'$ as required in the pattern cannot exist for
the following reasons:
\begin{itemize}
\item The state $q$ can only be reached via calls and therefore
  $(q,\bot)$ is not reachable from $(q,\bot)$.
\item From $q'$ the symbol $Z'$ is pushed onto the stack. But $q'$ can
  only be reached on popping $Z$. Therefore $(q',\bot)$ is not
  reachable from $(q',\bot)$.
\end{itemize}
\begin{figure}
\begin{center}
\begin{tikzpicture}[node distance=2cm,initial text=,scale=.5]
\tikzstyle{every state}=[inner sep=1pt,minimum size=8mm]
\node[state](q) {$q$};
\node[state,accepting](q'') [above right of=q]{$q''$};
\node[state](q') [below right of=q'']{$q'$};
\path[->]
          (q) edge[bend left] node[left]{$\scriptstyle c/Z$} (q'')
          (q'') edge[bend left] node[right]{$\scriptstyle c/Z$} (q)
               edge[loop above] node[above]{$\scriptstyle r_1/Z,r_1/Z'$} ()
               edge node[right]{$\scriptstyle r_2/Z$} (q')
          (q') edge node[below]{$\scriptstyle c/Z'$} (q)
;
\end{tikzpicture}
\end{center}
\caption{A stair B\"uchi DVPA illustrating the definition of forbidden
pattern}\label{fig:DVPA-with-forbidden-pattern}
\end{figure}
However, the example automaton $\A$ contains an extended pattern that 
guarantees that $\lsu \wred L_\omega(\A)$, as defined below and
illustrated in Figure~\ref{fig:forbidden}.
\begin{figure}
\begin{center}
\begin{tikzpicture}
\path
  node(1){$q$}
  -- ++(1.5,1) node(2){$q$}  +(0,-0.7) node[stack]{$\sigma$}
  -- ++(1.5,0) node(3){$q'$}  +(0,-0.7) node[stack]{$\sigma$}
  -- ++(1.5,1) node(4){$q$}  +(0,-1.7) node[stack]{$\sigma$}  
                             +(0,-0.7) node[stack]{$\sigma'$}
  -- ++(1.5,0) node(5){$q''$}  +(0,-1.7) node[stack]{$\sigma$}
                             +(0,-0.7) node[stack]{$\sigma'$}
  -- ++(1.5,-1) node(6){$q''$}  +(0,-0.7) node[stack]{$\sigma$}
  -- ++(1.5,-1) node(7){$q'$}
;
\path[->] 
(1) edge node[above left]{$u$} node[xshift=4,yshift=-4]{$\scriptstyle F$} (2)
 (2) edge[bend left] node[above]{$v$} node[below]{$\scriptstyle \notin F$} (3)
 (3) edge node[above left]{$w$} node[xshift=5,yshift=-5]{$\scriptstyle
  \notin F$} (4)
 (4) edge[bend left] node[above]{$x$} (5)
 (5) edge node[above right]{$y$} (6)
 (6) edge node[above right]{$z$} (7)
;
\end{tikzpicture}
\end{center}
\caption{Forbidden pattern}\label{fig:forbidden}
\end{figure}

Formally, we call $q,q' \in Q \setminus F$, $q'' \in Q$, $u,v,w,x,y,z
\in \alphabet^*$, and $\sigma, \sigma' \in \Gamma^*$ a forbidden pattern
of $\A$ if $uvwxyz \in \lmwm$ and
\[
\begin{array}{lll}
(q,\bot) \xrightarrow[F]{u} (q,\sigma), &
  (q,\bot) \xrightarrow[\notin F]{v} (q',\bot), & 
  (q',\bot) \xrightarrow[\notin F]{w} (q,\sigma'), 
\\ 
(q,\bot) \xrightarrow{x} (q'',\bot), &
(q'',\sigma') \xrightarrow{y} (q'', \bot), & 
(q'',\sigma) \xrightarrow{z} (q',\bot).
\end{array}
\]
Note that $\sigma'$ might be empty. Since $q$ is a non-final
state, and we require that a final state is seen on a step on the path
from $q$ to $q$, the stack content $\sigma$ cannot be empty. Further
note that this pattern subsumes the first simple pattern:  choose
$q=q'=q''$, $v=w=x=y=\bot$, and $u' = z$.

The example automaton from
Figure~\ref{fig:DVPA-with-forbidden-pattern} contains such a pattern
for $q,q',q''$. the words $u=cc$, $v=cr_2$, $w=c$, $x=cr_1$,
$y=r_1$, $z=r_1r_2$, and the stack contents $\sigma = ZZ$, $\sigma' =
Z'$. 

\begin{lemma} \label{lem:forbidden-pattern}
If $\A$ has a forbidden pattern, then $\lsu \wred L_\omega(\A)$.
\end{lemma}
\begin{proof}
We describe a winning strategy $f$ for Player~II in the Wadge
game. The basic idea is to play $u$ whenever Player~I plays $c$, and
to match the last open $u$ with $z$ whenever Player~I plays
$r$. However, after playing $z$, the automaton $\A$ is in state $q'$
(compare Figure~\ref{fig:forbidden}). Hence, to play $u$ again, we
first have to play $w$ to reach $q$, producing a $\sigma'$ on the
stack. Therefore, it can happen that we first have remove these
$\sigma'$ from the stack before we can match the last open $u$ with
$z$. To keep track of this, we use words over $\{0,1\}$ as memory for
$f$ representing an abstraction of the stack of $\A$ ($0$ corresponds
to $\sigma$ and $1$ corresponds to~$\sigma'$).

To simplify the description of $f$, we construct the moves such that
$\A$ is always in $q'$ after reading a finite word generated by
$f$. We also assume that $q'$ is the initial state of $\A$. If this is
not the case, Player~II can simply prepend to the first move a word
leading $\A$ to state $q'$.

Let $\eta \in \{0,1\}^*$ be the current memory content (the initial
content being $\varepsilon$). Then the strategy $f$ works as follows:
\begin{itemize}
\item If Player~I plays $c$, then play $wuv$ and update the memory to
  $01\eta$.
\item If Player~I plays $r$, then let $i \ge 0$ be such that $\eta$ is
  of the form $1^i0\eta'$. In this case, play $wxyy^iz$ and update the
  memory to $\eta'$.
\end{itemize}
Let $|\eta|_0$ denote the number of $0$ occurring in $\eta$ and let
$k$ be the number of final states seen on steps in the run
$(q,\bot) \xrightarrow{u} (q',\sigma)$. Note that $k \ge 1$ by
definition of forbidden pattern. By induction one shows that
\begin{enumerate}
\item after each move of Player~II the number of open calls in the
  word played by Player~I corresponds to $|\eta|_0$,
\item the number of final states seen on steps when $\A$ reads a finite
  word produced by $f$ is  $k \cdot |\eta|_0$.
\end{enumerate}
This implies that $\A$ accepts the infinite word produced by Player~II
according to $f$ iff the infinite word produced by Player~I contains
an unbounded number of unmatched calls. 
\end{proof}

%------------------------------------------------------------
\paragraph{Complexity of state pairs.}
%------------------------------------------------------------
We now show that the absence of forbidden patterns allows to construct
a parity {DVPA} $\A'$ that is equivalent to $\A$.  In order to find an
upper bound on the number of required priorities, we start by defining
a measure for the complexity of pairs of non-final states.  The pair
$(q,q')$ from Figure~\ref{fig:forbidden} would be of infinite
complexity.  If we now replace the states $q$ and $q'$ in the upper
part of Figure~\ref{fig:forbidden} by states $p$ and $p'$, then this
indicates that the possible runs between $q$ and $q'$ are at least as
complex as those between $p$ and $p'$. This situation is shown in
Figure~\ref{fig:order}. Since $q''$ is just an auxiliary state and not
of particular importance, we replaced it by $p''$ to obtain a more
consistent naming scheme.  We show that this relation indeed defines a
strict partial order on pairs of non-final states in the case that
$\A$ does not contain forbidden patterns.

For $p,p',q,q' \in Q \setminus F$ define $(p,p') \Aless (q,q')$ iff
there exists $p'' \in Q$ and stack contents $\sigma, \sigma'$ such
that (see Figure~\ref{fig:order} for an illustration):

\[
\begin{array}{lll}
(q,\bot) \xrightarrow[F]{u} (p,\sigma), &
  (p,\bot) \xrightarrow[\notin F]{} (p',\bot), &
  (p',\bot) \xrightarrow[\notin F]{} (p,\sigma'), \\
(p,\bot) \xrightarrow{} (p'',\bot), &
  (p'',\sigma') \xrightarrow{} (p'', \bot), &
  (p'',\sigma) \xrightarrow{z} (q',\bot),
\end{array}
\]
and $uz \in \lmwm$. The words $v,w,x,y$ from the definition of
forbidden pattern are not made explicit in this definition because we
never need to refer to them. As for forbidden patterns, $\sigma'$
might be empty but $\sigma$ must be non-empty.
\begin{figure}
\begin{center}
\begin{tikzpicture}
\path
  node(1){$q$}
  -- ++(1.5,1) node(2){$p$}  +(0,-0.7) node[stack]{$\sigma$}
  -- ++(1.5,0) node(3){$p'$}  +(0,-0.7) node[stack]{$\sigma$}
  -- ++(1.5,1) node(4){$p$}  +(0,-1.7) node[stack]{$\sigma$}  
                             +(0,-0.7) node[stack]{$\sigma'$}
  -- ++(1.5,0) node(5){$p''$}  +(0,-1.7) node[stack]{$\sigma$}
                             +(0,-0.7) node[stack]{$\sigma'$}
  -- ++(1.5,-1) node(6){$p''$}  +(0,-0.7) node[stack]{$\sigma$}
  -- ++(1.5,-1) node(7){$q'$}
;
\path[->] 
(1) edge node[above left]{$u$} node[xshift=4,yshift=-4]{$\scriptstyle F$} (2)
 (2) edge[bend left] node[above]{} node[below]{$\scriptstyle \notin F$} (3)
 (3) edge node[above left]{} node[xshift=5,yshift=-5]{$\scriptstyle
  \notin F$} (4)
 (4) edge[bend left] node[above]{} (5)
 (5) edge node[above right]{} (6)
 (6) edge node[above right]{$z$} (7)
;
\end{tikzpicture}
\end{center}
\caption{The relation $(p,p') \Aless (q,q')$}\label{fig:order}
\end{figure}
%
%\begin{figure}
%\begin{center}
%\begin{picture}(70,40)
%\gasset{Nframe=n,Nw=5,Nh=5}
%\node(1)(0,5){$q$}
%\node(2)(15,20){$p$}
%\node(3)(25,20){$p'$}
%\node(4)(35,30){$p$}
%\node(5)(45,30){$p''$}
%\node(6)(55,20){$p''$}
%\node(7)(65,5){$q'$}
%
%\gasset{Nframe=y,Nmr=0,Nw=3}
%\node[Nh=12](s2)(15,10){$\sigma$}
%\node[Nh=12](s3)(25,10){$\sigma$}
%\node[Nh=12](s4)(35,10){$\sigma$}
%\node[Nh=12](s5)(45,10){$\sigma$}
%\node[Nh=12](s6)(55,10){$\sigma$}
%\node[Nh=10](sp4)(35,23){$\sigma'$}
%\node[Nh=10](sp5)(45,23){$\sigma'$}
%
%\drawedge(1,2){$F$}
%\drawedge[ELside=r,AHnb=0,linewidth=0](1,2){}
%\drawedge[ELside=r,curvedepth=3](2,3){}
%\drawedge[AHnb=0,linewidth=0,curvedepth=3](2,3){$\notin F$}
%\drawedge[ELside=r](3,4){}
%\drawedge[AHnb=0,linewidth=0](3,4){$\notin F$}
%\drawedge[curvedepth=3](4,5){}
%\drawedge(5,6){}
%\drawedge(6,7){}
%\end{picture}
%\end{center}
%\caption{The relation $(p,p') \Aless (q,q')$}\label{fig:order}
%\end{figure}
%
\begin{lemma}
If $\A$ does not have a forbidden pattern, then $\Aless$ is a strict
partial order on pairs of states.
\end{lemma}
\begin{proof}
We have to show that $\Aless$ is transitive and irreflexive (asymmetry
follows from these two). The relation is obviously irreflexive because
of the absence of forbidden patterns. Transitivity is illustrated in
Figure~\ref{fig:transitive} for $(r,r') \Aless (p,p') \Aless (q,q')$
(the stack contents are omitted). The shown pattern is obtained from
$(r,r') \Aless (p,p') \Aless (q,q')$. The configurations with a frame
lead to a pattern witnessing $(r,r') \Aless (q,q')$.
\end{proof}
\begin{figure}
\begin{center}
\begin{tikzpicture}
\path
  node[draw](1){$q$}
  -- ++(1,1) node(2){$p$}   

  -- ++(1,1) node[draw](2r){$r$}   
  -- ++(1,0) node[draw](3r){$r'$}  
  -- ++(1,1) node[draw](4r){$r$}   
  -- ++(1,0) node[draw](5r){$r''$} 
  -- ++(1,-1) node[draw](6r){$r''$}

  -- ++(1,-1) node(3){$p'$}
  -- ++(1,1) node(4){$p$}   
  -- ++(1,0) node(5){$p''$} 
  -- ++(1,-1) node(6){$p''$}
  -- ++(1,-1) node[draw](7){$q'$}
;
\path[->] 
(1) edge node[above left]{} node[xshift=4,yshift=-4]{$\scriptstyle F$} (2)

 (2) edge node[above left]{} node[xshift=4,yshift=-4]{$\scriptstyle F$} (2r)
 (2r) edge[bend left] node[above]{} node[below]{$\scriptstyle \notin F$} (3r)
 (3r) edge node[above left]{} node[xshift=5,yshift=-5]{$\scriptstyle
  \notin F$} (4r)
 (4r) edge[bend left] node[above]{} (5r)
 (5r) edge node[above right]{} (6r)
 (6r) edge node[above right]{} (3)

 (3) edge node[above left]{} node[xshift=5,yshift=-5]{$\scriptstyle
  \notin F$} (4)
 (4) edge[bend left] node[above]{} (5)
 (5) edge node[above right]{} (6)
 (6) edge node[above right]{} (7)

;
\end{tikzpicture}
\end{center}
\caption{Transitivity of $\Aless$}\label{fig:transitive}
\end{figure}
For $\A$ without forbidden patterns, we assign to each pair of states
a number according to its height in the partial order, i.e., $\prio:
Q^2 \rightarrow \nat$ is a mapping satisfying
\[
\prio(q,q') = \max(\{0\} \cup \{\prio(p,p') \mid (p,p') \Aless
(q,q')\}) + 1.
\]
We need the following simple observation.
\begin{lemma} \label{lem:reachable-height}
Let $q_1,q_1',q_2,q_2' \in Q \setminus F$. If there is a stack content
$\sigma$ such that $(q_2,\bot) \xrightarrow{u} (q_1,\sigma)$
and $(q_1',\sigma) \xrightarrow{v} (q_2',\bot)$ with $uv \in
\lmwm$, then $\prio(q_2,q_2') \ge \prio(q_1,q_1')$.
\end{lemma}
\begin{proof}
The condition $(q_2,\bot) \xrightarrow{u} (q_1,\sigma)$ and
$(q_1',\sigma) \xrightarrow{v} (q_2',\bot)$ with $uv \in \lmwm$
implies that whenever $(q,q') \Aless (q_1,q_1')$, then also $(q,q')
\Aless (q_2,q_2')$. Thus, $\prio(q_2,q_2') \ge \prio(q_1,q_1')$ by
definition of $\prio$. 
\end{proof}

To make use of $\Aless$ and $\prio$ in the construction of $\A'$ we
need the following lemma. Note that this statement does not assume
that $\A$ as no forbidden patterns. 
\begin{lemma} \label{lem:polynomial}
The relation $\Aless \subseteq (Q \setminus F)^2$ can be computed in
time polynomial in the size of $\A$.
\end{lemma}
\begin{proof}
In \cite{EsparzaHRS00} it is shown that for a given configuration $p\sigma$
of $\A$ one can compute in polynomial time the set
$\textit{pre}^*(q\sigma)$ of configurations from which there is a run to
$p\sigma$, and the set $\textit{post}^*(q\sigma)$ of configurations that are
reachable from $p\sigma$ by a run. These sets of configurations are sets of
words over $\Gamma$, starting with a symbol from $Q$, and can be
represented by finite automata.

The algorithms from \cite{EsparzaHRS00} can be modified to consider
only runs that either see a final state on a step or do not see a
final state on a step, resulting in the sets  $\textit{pre}^*_F(q\sigma)$,
$\textit{pre}^*_{\notin F}(q\sigma)$, and similarly for $\textit{post}$.

For checking whether $(p,p') \Aless (q,q')$ it is sufficient to check
for each $p''$ if there are runs as required in the definition of
$\Aless$. This can be done by a suitable combination of the above
mentioned algorithms. For example, the stack content $\sigma$ would be
obtained by finding a $\sigma$ such that $p\sigma \in
\textit{post}^*_F(q\bot)$, and $p''\sigma
\in \textit{pre}^*(q'\bot)$. Similarly for $\sigma'$. 

All these computations can be done in polynomial time, and there are
only polynomially many combinations of states that have to tested.
\end{proof}

%------------------------------------------------------------
\paragraph{Informal description of the parity {DVPA}.}
%------------------------------------------------------------
In a B\"uchi stair condition, a final state visited in a run is
``erased'' (in the sense that it is not considered for acceptance), if
it is not on a step. If we construct a parity DVPA, then we cannot
erase states like this. Instead, we use the mechanisms of different
priorities to simulate erasing a state. Roughly, final states of the
stair B\"uchi automaton are translated into even priorities. If a
final state is erased, then this is compensated by visiting a higher
odd priority. For the choice of the correct priorities we use the
function $\prio$. 

In the description below, we use the terminology of ``$\A$ closing a
pair $(q,q')$ of states''. This means that $\A$ was in state $q$ at
some position and after reading a word $\lmwm$ it reached state $q'$,
i.e., $\A$ was in state $q$ before reading a call and reached $q'$
after the matching return.

As mentioned above, we somehow need to determine a priority for the
final states that are visited. Assume that the automaton is in
configuration $(q,\beta)$ and reads a word that increases the stack
height leading to some configuration $(p,\sigma\beta)$ and visiting
some final states on steps during this run. We do not know if these
final states remain on steps or will be erased at some point. But if
we knew, e.g., that whenever we come back to the stack content $\beta$
with, say, state $q'$, that the pair $(q,q')$ is of height at least
$i$, then we could signal priority $2i$ for the final states that we
have seen after $(q,\beta)$ and signal priority $2i+1$ if we indeed
close a pair $(q,q')$ on the level of $\beta$, and thus erasing all
the final states.

Assume that we have already seen the pattern shown in
Figure~\ref{fig:priority-pattern}, where $(p,p')$ is a pair of height
$i-1$. Then $\prio(q,q') \ge i$ for every state $q'$ that we could
reach when coming back to the stack height of the configuration with
$q$ at the beginning of this pattern.
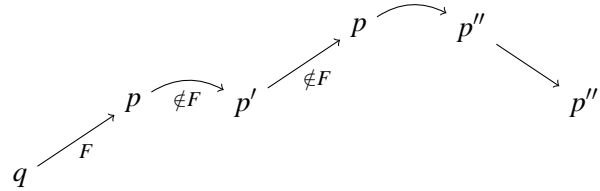
\begin{figure}
\begin{center}
\begin{tikzpicture}
\path
  node(1){$q$}
  -- ++(1.5,1) node(2){$p$}   
  -- ++(1.5,0) node(3){$p'$}  
  -- ++(1.5,1) node(4){$p$}   
                              
  -- ++(1.5,0) node(5){$p''$} 
                              
  -- ++(1.5,-1) node(6){$p''$}
;
\path[->] 
(1) edge node[above left]{} node[xshift=4,yshift=-4]{$\scriptstyle F$} (2)
 (2) edge[bend left] node[above]{} node[below]{$\scriptstyle \notin F$} (3)
 (3) edge node[above left]{} node[xshift=5,yshift=-5]{$\scriptstyle
  \notin F$} (4)
 (4) edge[bend left] node[above]{} (5)
 (5) edge node[above right]{} (6)
;
\end{tikzpicture}
\end{center}
\caption{The pattern for determining the priority of the
  states with $\prio(p,p') = i$}\label{fig:priority-pattern}
\end{figure}
In particular, if $h$ is the maximal height of a pair of states, and
$(p,p')$ are of height $h$, then we know that the final states between
$q$ and $p$ cannot all be deleted because this would require closing a
pair of height $h+1$. 

By a simple combinatorial argument, one can see that such a pattern as
shown in Figure~\ref{fig:priority-pattern} must occur if $\A$, before
returning to the stack height of $q$, has successively closed $m :=
|Q|^3 + 1$ pairs $(p_1.p_1'), \ldots, (p_m.p_m')$ of height $i-1$
without visiting final states on steps in between, as illustrated in
Figure~\ref{fig:detecting-height} (in the picture the pairs are closed
on increasing stack levels, however, they can also be on the same
stack level). If we denote by $p_i''$ the states of $\A$ the next time
it reaches the stack level of $(p_i,p_i')$ (indicated by the dotted
line in the picture), then one such triple of states must occur twice,
giving rise to a pattern witnessing that $\prio(q,q') \ge i$.

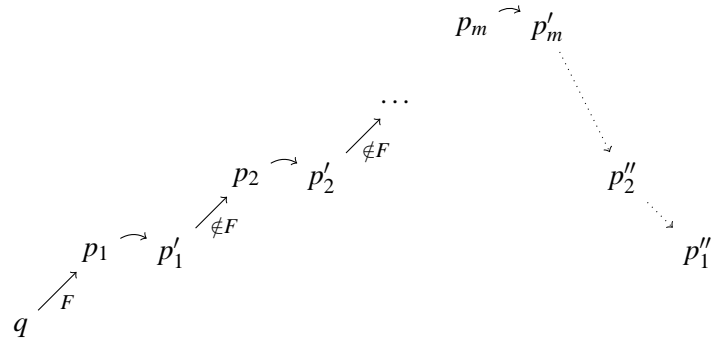
\begin{figure}
\begin{center}
\begin{tikzpicture}
\path
  node(1){$q$}
  -- ++(1,1) node(2){$p_1$}   
  -- ++(1,0) node(3){$p_1'$}  
  -- ++(1,1) node(4){$p_2$}   
  -- ++(1,0) node(5){$p_2'$}  
  -- ++(1,1) node(6){$\cdots$}   
  -- ++(1,1) node(7){$p_m$}   
  -- ++(1,0) node(8){$p_m'$}  
  -- ++(1,-2) node(9){$p_2''$}  
  -- ++(1,-1) node(10){$p_1''$}  
;
\path[->] 
(1) edge node[xshift=4,yshift=-4]{$\scriptstyle F$} (2)
 (2) edge[bend left] (3)
 (3) edge node[xshift=5,yshift=-5]{$\scriptstyle \notin F$} (4)
 (4) edge[bend left] (5)
 (5) edge node[xshift=5,yshift=-5]{$\scriptstyle \notin F$} (6)
 (7) edge[bend left] (8)
 (8) edge[dotted] (9)
 (9) edge[dotted] (10)
;
\end{tikzpicture}
\end{center}
\caption{Detecting that each pair with $q$ is of height at least $i$.} \label{fig:detecting-height}
\end{figure}

%\begin{center}
%{\psfrag{q}{$q$}
%\psfrag{Q3}{$|Q|^3+1$}
%\psfrag{F}{$F$}
%\epsfig{file=repetition.eps, width=5cm}
%}
%\end{center}

To detect such situations, $\A'$ maintains a counter with range from
$0$ to $m$ for each possible height of state pairs, and roughly
behaves as follows:
\begin{itemize}
\item Whenever a pair of height $i$ is closed by $\A$, then counter
  $i$ is increased by one (and for technical reasons counter number
  $0$ is increased whenever $\A$ visits a non-final state after
  reading a call or an internal symbol). To detect the closed pairs,
  $\A'$ stores the states of $\A$ on the stack, and the height of
  state pairs can be computed by Lemma~\ref{lem:polynomial}.
\item There is an additional flag for each $i \in \{0, \ldots, h\}$
  indicating whether counter number $i$ was reset because a final
  state of $\A$ has been visited (the flag is set to $1$), or because
  it reached its maximal value $m$ (the flag is set to $0$).
\item When counter number $i$ reaches value $m$ (if several counters
  reach $m$ at the same time we take the maximal such $i$), then the
  automaton signals priority $2i+2$ if the flag number $i$ is set, and
  $2i+1$ if the flag is not set. In the next transition the counter is
  reset.
\end{itemize}

%------------------------------------------------------------
\paragraph{Formal description of the parity DVPA.}
%------------------------------------------------------------
\newcommand{\stack}[1]{[ #1 ]}
\newcommand{\push}[2]{\begin{array}{c}#1 \\ #2 \end{array}}
Recall that $m := |Q|^3+1$ and that $h$ is the maximal height of a
pair of states from $Q \setminus F$.

\begin{itemize}
\item The states of $\A'$ are of the form $(q,\chi,f)$, where $q \in Q$ is a
state of $\A$, $\chi: \{0, \ldots, h\} \rightarrow \{0, \ldots, m\}$
represents the counters mentioned above, and $f: \{0, \ldots, h\}
\rightarrow \{0,1\}$ represents the flag mentioned in the informal
description.

\item The stack symbols of $\A'$ are of the form
$\stack{Z,(q,\chi,f)}$, where $Z$ is a stack symbol of $\A$
and $(q,\chi,f)$ is a state of $\A'$.

\item We now define when $\A'$ can move from state $(q,\chi,f)$ to state
$(q',\chi',f')$, depending on whether it reads a call, an internal
action, or a return. In all cases, $q'$ is the next state of $\A$,
i.e., $\A'$ simulates $\A$ in its first component.  If $q' \in F$,
then $\chi' = 0$ and $f' = 1$, i.e., the constant functions mapping
everything to $0$ and $1$, respectively. The other cases for $\delta'$
are listed below:
\begin{description}
\item[Call:] $(q,\chi,f) \xrightarrow{c}
  \push{(q',\chi',f')}{\stack{Z,(q,\chi,f)}}$ if $\delta(q,c) =
  (Z,q')$, $q' \notin F$, and
\[
\chi'(i) = \left\{
\begin{array}{l}
(\chi(i) \mod m) + 1 \mbox{ if } i = 0, \\
(\chi(i) \mod m)  \mbox{ otherwise, }
\end{array} \right.
\hspace{.5cm}
f'(i) = \left\{
\begin{array}{l}
f(i) \mbox{ if } \chi(i) < m, \\
0  \mbox{ otherwise. }
\end{array} \right.
\]

%$\delta'((q,\chi,f),c) = ((q',\chi',f'),[Z,(q,\chi,f)]) \xrightarrow{c} $
\item[Internal action:] $(q,\chi,f) \xrightarrow{a} (q',\chi',f')$ if
  $\delta(q,a) = q'$, $q' \notin F$, and $\chi'$ and $f'$ are as in
  the case of a call symbol.
\item[Return:] $\push{(q,\chi,f)}{\stack{Z,(q'',\chi'',f'')}}
\xrightarrow{r} (q',\chi',f')$ if $\delta(q,Z,r) = q'$, $q'
\notin F$, and
\[
\begin{array}{l}
{
\chi'(i) = \left\{
\begin{array}{l}
(\chi''(i) \mod m) + 1 \mbox{ if } q'' \notin F \mbox{ and } i \le \prio(q'',q'), \\
(\chi''(i) \mod m)  \mbox{ otherwise, }
\end{array} \right.
}
\\ \\
{
f'(i) = \left\{
\begin{array}{l}
f''(i) \mbox{ if } \chi''(i) < m, \\
0  \mbox{ otherwise. }
\end{array} \right.
}
\end{array}
\]
\end{description}
\item The priority function $\Omega'$ of $\A'$ is defined as follows
\[
\Omega'(q,\chi,f) = 
\left\{
\begin{array}{l}
0 \mbox{ if } \chi(i) < m \mbox{ for all } i,\\
2d + 1 + f(d) \mbox{ if }  d = \max\{i \mid \chi(i) = m\}.
\end{array} \right.
\]
\item The initial state is $(q_0,\chi_0,f_0)$ with $\chi_0 = 0$ and
  $f_0 = 1$.
\end{itemize}

\begin{lemma} \label{lem:DVPAconstruction}
The parity {DVPA} $\A'$ is equivalent to $\A$.
\end{lemma}
\begin{proof}
We note the following helpful fact on reachable states $(q,\chi,f)$ of $\A'$: 
\begin{enumerate}[(1)]
\item If $f(i) = 1$ for some $i$, then $f(j) = 1$ and $\chi(i) \ge \chi(j)$
 for all $j \ge i$. The initial state satisfies this property, and if
 we apply the definition of the transition function to a state
 satisfying the property, then one can easily verify that the
 resulting state also satisfies it.
\end{enumerate}
Now consider an accepting run of $\A$. We show that the corresponding
run of $\A'$ is also accepting.  Let the $k$th state in this run of
$\A'$ be $(q_k,\chi_k,f_k)$.

If $\ell$ is a step in the run and $q_{\ell}$ is a final state of
$\A$, then all flags are set to 1 at this point. From the definition
of $\delta'$ follows that these flags can only be set to $0$ if the
corresponding counter reaches value $m$ (we assume that the final
state occurs on a step and therefore the run never accesses the stack
symbols below). Now assume that $\A'$ signals some odd priority $2i+1$
at some position $k$ after this final state. This means that $i$ is
maximal with $\chi_k(i) = m$, and furthermore $f_k(i) = 0$. But if
$f_k(i) = 0$, then there must be some $k'$ with $\ell < k' < k$ such
that $f_{k'}(i) = 1$ and $\chi_{k'}(i) = m$ because this is the only
situation in which the flag is set to $0$.

From (1) we conclude that $f_{k'}(j) = 1$ for all $j \ge i$ and hence
$\Omega'(q_{k'},\chi_{k'},f_{k'})$ is an even priority bigger than
$2i+1$. Thus, for each odd priority occurring after a final state on a
step there is a bigger even priority also occurring after this final
state.  Hence, the run of $\A'$ is also accepting.

For the other direction, consider a non-accepting run of $\A$ and as
before let $(q_k,\chi_k,f_k)$ be the $k$th state in the corresponding
run of $\A'$. There is a position such that after this position no
final states of $\A$ occur on a step. From now on we only consider
this part of the run.  

Consider the sequence $k_1,k_2,k_3, \ldots$ of steps. As no final
state occurs on a step we have the following relation between the
counter values at two successive steps:
\begin{enumerate}[(i)]
\item If $k_{j+1}$ was reached from $k_j$ by reading a call or an
  internal symbol, then the only change of the counters is
  $\chi_{k_{j+1}}(0) = (\chi_{k_{j}}(0) \mod m)+1$. The other values
  remain the same.
\item If $k_{j+1}$ was reached from $k_j$ by reading a minimally
  well-matched word, then the counters are updated as follows:
\[
\chi_{k_{j+1}}(i) = \left\{
\begin{array}{l}
(\chi_{k_j}(i) \mod m) + 1 \mbox{ if } i \le \prio(q_{k_j},q_{k_{j+1}}), \\
(\chi_{k_j}(i) \mod m)  \mbox{ otherwise. }
\end{array} \right.
\]
\end{enumerate}
The flags between two successive steps are updated as follows:
\[
f_{k_{j+1}}(i) = \left\{
\begin{array}{l}
f_{k_j}(i) \mbox{ if } \chi_{k_j}(i) < m, \\
0  \mbox{ otherwise. }
\end{array} \right.
\]
Now let $d$ be the highest counter that is infinitely often increased
on a step (such a counter exists because counter $0$ is increased for
each call and each internal symbol). Then the highest priority
occurring on a step is obviously $2d+1$ because after the first reset
of counter $d$ to $0$ the flag number $d$ is $0$ on all following
steps. 

We have to show that no even priority higher than $2d+1$ can occur
infinitely often. Restrict the part of the run under consideration
further to the suffix on which no counter higher than $d$ is
incremented on a step. We can conclude that for successive steps
connected by a minimally well-matched word we have that
$\prio(q_{k_j},q_{k_{j+1}}) \le d$.

We first assume that $d > 0$. At the end of the proof we briefly
explain the case $d = 0$.

Pick $j$ such that there is $\ell$ with $k_j < \ell < k_{j+1}$ and
$\Omega'(q_\ell, \chi_\ell, f_\ell) = 2i+2$ (if no such position
exists, then the run of $\A'$ is clearly rejecting). For simplicity
let $(q_{k_j}, \chi_{k_j}, f_{k_j}) = (q,\chi,f)$ and $(q_{k_{j+1}},
\chi_{k_{j+1}}, f_{k_{j+1}}) = (q',\chi',f')$. 

We now consider the part of the run from $k_j$ to $\ell$ and show that
$i < \prio(q,q') \le d$ and hence $2i+2 < 2d+1$.

Since $\Omega'(q_\ell, \chi_\ell, f_\ell) = 2i+2$ we know that
$f_\ell(i) = 1$ and $i$ is maximal with $\chi_\ell(i) = m$. If $i = 0$
we know that $i < d$ by our assumption $d > 0$.  If $i > 0$, at
position $\ell$ a pair of states of height $i$ is closed.  From
Lemma~\ref{lem:reachable-height} we obtain that $d \ge \prio(q,q') \ge
i$.

There are two cases to consider.  If flag number $i$ was already set
to $1$ at position $k_j$, i.e., $f(i) = 1$, then $i \not= d$ (as we
only consider the part of the run where the flag for $d$ remains 0
forever on the steps).  Together with $d \ge i$ we get $d > i$.

If $f(i) = 0$, then it must be reset to $1$ by visiting a final
state. At the same time the counters are reset to $0$. Then $m$ pairs
of height $i$ have to be closed to reach the value $\chi_\ell(i) =
m$. Furthermore, these pairs have to closed at positions that
correspond to steps in the part of the run between $k_j$ and $\ell$
(not steps in the whole run). Let these pairs be $(p_1,p_1'),
(p_2,p_2'), \ldots, (p_m,p_m')$ (see
Figure~\ref{fig:detecting-height}) and the corresponding pairs of
positions be $(\ell_1,\ell_1') \ldots, (\ell_m,\ell_m')$. Now consider
for each $n$ the minimal position $\ell_n''$ with $\ell \le \ell_n''
\le k_{j+1}$ such that the stack height at $\ell_n'$ and $\ell_n''$ is
the same. Let $p_n''$ denote the state at the corresponding
position. By the choice of $m$ we get that there are $n_1 \not= n_2$
such that $(p_{n_1},p_{n_1}',p_{n_1}'') =
(p_{n_2},p_{n_2}',p_{n_2}'')$. Denote the corresponding triple by
$(p,p',p'')$. This triple witnesses that $\prio(q,q') > \prio(p,p') =
i$ as illustrated in the following picture:
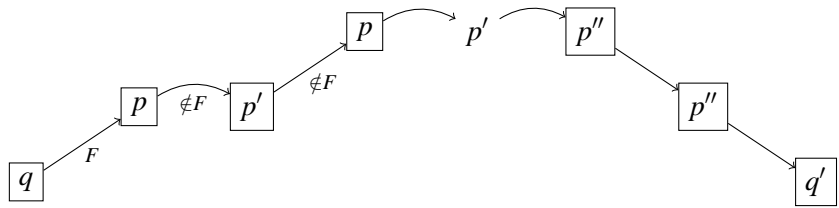
\begin{center}
\begin{tikzpicture}
\path
  node[draw](1){$q$}
  -- ++(1.5,1)  node[draw](2){$p$}   
  -- ++(1.5,0)  node[draw](3){$p'$}  
  -- ++(1.5,1)  node[draw](4){$p$}    
  -- ++(1.5,0)  node(4'){$p'$}       
  -- ++(1.5,0)  node[draw](5){$p''$} 
  -- ++(1.5,-1) node[draw](6){$p''$} 
  -- ++(1.5,-1) node[draw](7){$q'$}
;
\path[->] 
(1) edge node[above left]{} node[xshift=4,yshift=-4]{$\scriptstyle F$} (2)
 (2) edge[bend left] node[above]{} node[below]{$\scriptstyle \notin F$} (3)
 (3) edge node[above left]{} node[xshift=5,yshift=-5]{$\scriptstyle
  \notin F$} (4)
 (4) edge[bend left] node[above]{} (4')
 (4') edge[bend left] node[above]{} (5)
 (5) edge node[above right]{} (6)
 (6) edge node[above right]{} (7)
;
\end{tikzpicture}
\end{center}
It remains to consider the case $d = 0$. Consider only the suffix of
the run after the position where the flag for counter $0$ remains $0$
on all steps and no other counter is increased on a step anymore. Then
all pairs closed on steps are of height $0$ and by
Lemma~\ref{lem:reachable-height} pairs closed between two successive
steps are also of height $0$. So the maximal priority that we can see
on this part of the run would be $2$. For this to happen, the flag for
counter $0$ must be $1$ and counter $0$ must have value $m$. The flags
are only set to $1$ if a final state of $\A$ is reached, and at the
same time the counters are set to $0$.  Let $q,q'$ be the states at
two successive steps, and assume that in between a final state is
seen. Let $p$ be the state after the symbol following the final
state. If this symbol is a call or an internal, then $(p,p) \Aless
(q,q')$ (choosing $p'' = p$), contradicting $\prio(q,q') = 0$. Thus,
each final state of $\A$ is immediately followed by a return. Thus,
whenever the flag is set to $1$ by a final state, it is immediately
reset to $0$ in the next transition, and thus priority $2$ never
occurs (on the considered part of the run).
\end{proof}

Combining Lemmas~\ref{lem:forbidden-pattern} and
\ref{lem:DVPAconstruction} we obtain the following. 
\begin{theorem}
A stair B\"uchi DVPA $\A$ is equivalent to a parity DVPA if, and only
if, it does not contain any forbidden patterns. 
\end{theorem}
The relation $\Aless$ can be computed and checked for irreflexivity in
polynomial time. Hence we get the following corollary.
\begin{corollary}
For a stair B\"uchi DVPA $\A$ it is decidable in polynomial time if it is
equivalent to some parity DVPA.
\end{corollary}
A direct consequence of Lemma~\ref{lem:DVPAconstruction} is:
\begin{theorem}
If a stair B\"uchi DVPA $\A$ is equivalent to some parity DVPA, then
we can effectively construct such a parity DVPA.
\end{theorem}
It seems possible to lift the methods presented in this section to
decide for general stair parity DVPAs whether the stair condition is
required. We have, however, not yet worked out the details. A simpler
question can be solved using the game theoretic approach for deciding
the parity index problem for DPDAs: Given a stair parity DVPA $\A$ and
a set $P$ of priorities, we can decide whether there is a parity DVPA
using the priorities from $P$ that accepts $L_\omega(\A)$ by using the
classification game. In this case, the classification game could be
formalized using a combination of a classical parity and a stair
parity condition. Pushdown games with such a winning condition can be
solved with the methods from \cite{LMS2004VPG}.

%%%%%%%%%%%%%%%%%%%%%%%%%%%%%%%%%%%%%%%%%%%%%%%%%%%%%%%%%%%%%%%%%%%%%%%%%%%%%%%%
%%%%%%%%%%%%%%%%%%%%%%%%%%%%%%%%%%%%%%%%%%%%%%%%%%%%%%%%%%%%%%%%%%%%%%%%%%%%%%%%
\section{Conclusion} \label{sec:conclusion}
%%%%%%%%%%%%%%%%%%%%%%%%%%%%%%%%%%%%%%%%%%%%%%%%%%%%%%%%%%%%%%%%%%%%%%%%%%%%%%%%
%%%%%%%%%%%%%%%%%%%%%%%%%%%%%%%%%%%%%%%%%%%%%%%%%%%%%%%%%%%%%%%%%%%%%%%%%%%%%%%%
We have considered several decidability questions for
$\omega$-DPDAs. The regularity and equivalence problem are still open
for the full class of $\w$-DPDAs. We have sketched some partial
results from \cite{LodingR12} showing the decidability for these two
problems for the class of weak $\w$-DPDAs by a reduction to the
corresponding problems for DPDAs on finite words. It seems that
a decidability result for the full class of $\w$-DPDAs requires new
ideas.

In the second part we have analyzed the problem of simplifying the
acceptance condition of $\w$-DPDAs. We have shown that the smallest
number of priorities required for accepting the language of a given
parity DPDA can be computed. For the standard parity condition we have
used a game approach. For stair parity DVPAs, this problem can be
solved by a much simpler algorithm that uses a reduction to the
computation of the parity index of a finite automaton.

We have also shown that for stair B\"uchi DVPAs it is decidable
whether the stair condition is required or whether there exists an
equivalent parity DVPA. It seems that the methods used in the proof
can be generalized from stair B\"uchi conditions to arbitrary stair
parity conditions but we have not worked out the details.

%\nocite{*}
\bibliography{bib-afl}%\rand{Bibliography zu vereinheitlichen}
\bibliographystyle{eptcs}

\end{document}